\documentclass[journal,12pt, draftclsnofoot, onecolumn]{IEEEtran}
\usepackage{times,amssymb,amsmath,epsfig,nicefrac,euscript,cite,mathrsfs,color}

\newcommand\nc\newcommand
\nc\bfa{{\boldsymbol a}}\nc\bfA{{\boldsymbol A}}\nc\cA{{\mathcal A}}
\nc\bfb{{\boldsymbol b}}\nc\bfB{{\boldsymbol B}}\nc\cB{{\mathcal B}}
\nc\bfc{{\boldsymbol c}}\nc\bfC{{\boldsymbol C}}\nc\cC{{\mathcal C}}
\nc\sC{{\mathscr C}}
\nc\bfd{{\boldsymbol d}}\nc\bfD{{\boldsymbol D}}\nc\cD{{\mathcal D}}
\nc\bfe{{\boldsymbol e}}\nc\bfE{{\boldsymbol E}}\nc\cE{{\mathcal E}}
\nc\bff{{\boldsymbol f}}\nc\bfF{{\boldsymbol F}}\nc\cF{{\mathcal F}}
\nc\bfg{{\boldsymbol g}}\nc\bfG{{\boldsymbol G}}\nc\cG{{\mathcal G}}
\nc\bfh{{\boldsymbol h}}\nc\bfH{{\boldsymbol H}}\nc\cH{{\mathcal H}}
\nc\bfi{{\boldsymbol i}}\nc\bfI{{\boldsymbol I}}\nc\cI{{\mathcal I}}
\nc\bfj{{\boldsymbol j}}\nc\bfJ{{\boldsymbol J}}\nc\cJ{{\mathcal J}}
\nc\bfk{{\boldsymbol k}}\nc\bfK{{\boldsymbol K}}\nc\cK{{\mathcal K}}
\nc\bfl{{\boldsymbol l}}\nc\bfL{{\boldsymbol L}}\nc\cL{{\mathcal L}}
\nc\bfm{{\boldsymbol m}}\nc\bfM{{\boldsymbol M}}\nc\sM{{\mathscr M}}
\nc\bfn{{\boldsymbol n}}\nc\bfN{{\boldsymbol N}}\nc\cN{{\mathcal N}}
\nc\bfo{{\boldsymbol o}}\nc\bfO{{\boldsymbol O}}\nc\cO{{\mathcal O}}
\nc\bfp{{\boldsymbol p}}\nc\bfP{{\boldsymbol P}}\nc\cP{{\mathcal P}}
\nc\bfq{{\boldsymbol q}}\nc\bfQ{{\boldsymbol Q}}\nc\cQ{{\mathcal Q}}
\nc\bfr{{\boldsymbol r}}\nc\bfR{{\boldsymbol R}}\nc\cR{{\mathcal R}}
\nc\bfs{{\boldsymbol s}}\nc\bfS{{\boldsymbol S}}\nc\cS{{\mathcal S}}
\nc\bft{{\boldsymbol t}}\nc\bfT{{\boldsymbol T}}\nc\cT{{\mathcal T}}
\nc\bfu{{\boldsymbol u}}\nc\bfU{{\boldsymbol U}}\nc\cU{{\mathcal U}}
\nc\bfv{{\boldsymbol v}}\nc\bfV{{\boldsymbol V}}\nc\cV{{\mathcal V}}
\nc\bfw{{\boldsymbol w}}\nc\bfW{{\boldsymbol W}}\nc\cW{{\mathcal W}}
\nc\bfx{{\boldsymbol x}}\nc\bfX{{\boldsymbol X}}\nc\cX{{\mathcal X}}
\nc\bfy{{\boldsymbol y}}\nc\bfY{{\boldsymbol Y}}\nc\cY{{\mathcal Y}}
\nc\bfz{{\boldsymbol z}}\nc\bfZ{{\boldsymbol Z}}\nc\cZ{{\mathcal Z}}

\nc{\remove}[1]{}

\def\hq{\qopname\relax{no}{H_q}}

\def\avg{{\mathbb E}}

\newtheorem{theorem}{Theorem}
\newtheorem{definition}{Definition}
\newtheorem{lemma}{Lemma}
\newtheorem{proposition}[theorem]{Proposition}
\newtheorem{corollary}{Corollary}

\newtheorem{remark}{\indent Remark}

\newcommand\ff{{\mathbb F}}

\newcommand\integers{{\mathbb Z}}
\newcommand\define{{\stackrel{\triangle}{=}}}

\allowdisplaybreaks

\begin{document}
\title{An Upper Bound On the Size of Locally Recoverable Codes}
\author{Viveck Cadambe~\IEEEmembership{Member,~IEEE} and Arya Mazumdar~\IEEEmembership{Member,~IEEE}

\thanks{A preliminary version of this work has appeared in
  IEEE International Symposium on Network Coding (NetCod), 2013. This work was supported in
  part by  NSF CCF 1318093.}  \thanks{A.~Mazumdar is 
   with the Department of Electrical and Computer
  Engineering, University of Minnesota, Minneapolis, MN 55455 (Email:
  arya@umn.edu).} 
\thanks{Viveck Cadambe is with the Department of Electrical Engineering, 
  Pennsylvania State University, University Park PA 16802 (Email: viveck@engr.psu.edu).}}

\maketitle

\begin{abstract}
In a {\em locally recoverable} or {\em repairable} code,
any symbol of a codeword can be recovered by reading only a
small (constant) number of other symbols. The notion of local
recoverability is important in the area of distributed storage where
a most frequent error-event is a single storage node failure (erasure).  A common
objective is to repair the node by downloading data from as few
other storage node as possible. In this paper, we bound the minimum 
distance of a code in terms of  its length, size and locality. Unlike previous bounds, our 
bound follows from a significantly simple analysis and depends on the size of the alphabet being used. 
It turns out that the binary Simplex codes satisfy our bound with equality; hence the Simplex codes are the first example of
a optimal binary locally repairable code family.
We also provide achievability results based on random coding and concatenated codes that are numerically verified to be close to our bounds.
\end{abstract}

\section{Introduction}
The increased demand of cloud computing and storage services in current times has led to a corresponding surge in the study and deployment of erasure-correcting codes, or simply erasure codes, for distributed storage systems. In the information and coding theory community, this has led to the research of some new aspects of codes particularly tailored to the application to storage systems. The topic of interest of this paper is the locality of repair of erasure codes.

It is well known that an erasure code with length $n$, dimension $k$ and minimum distance $d$, or an $(n,k,d)$  code, can recover from \emph{any} set of $d-1$ erasures. In addition, the code is said to have \emph{locality} $r$ if any \emph{single} erasure can be recovered from some set of $r$ symbols of the codeword. From an engineering perspective, when an $(n,k,d)$ code is used to store information in $n$ storage nodes, the parameter $d$ represents the worst-case (node) failure scenario from which the storage system can recover. The parameter $r,$ on the other hand, represents the efficiency of recovery from a (relatively) more commonly occurring scenario - a single node failure. It is therefore desirable to have a large value of $d$ and a small value of $r$. Much literature in classical coding theory has been devoted to understanding the largest possible value of $d$ - the minimum distance - when the parameters $(n,k)$ are fixed; the well-known results from this body \cite{Roth_Book} of work include the Singleton bound, and code constructions that achieve this bound (such as Reed-Solomon codes). The study of minimizing the locality, $r$, was initiated recently in \cite{Gopalan_Locality, oggier2011self} and furthered in \cite{Dimakis_LRC,VijayKumar_MBRLocality,rawat_LRCoptimal, Dimakis_Facebook, Cheng_Azure_FAST, Dimitris_SimpleLRC}. The key discovery of \cite{Gopalan_Locality, Dimakis_LRC, prakash2012optimal} is that, for any $(n,k,d)$ code with locality $r,$ the following bound is satisfied:
\begin{equation} 
d \leq n-k -\lceil k/r\rceil + 2.
\label{eq:Yekhanin_etal}
\end{equation}
\begin{figure}
\center
\centerline{\psfig{figure=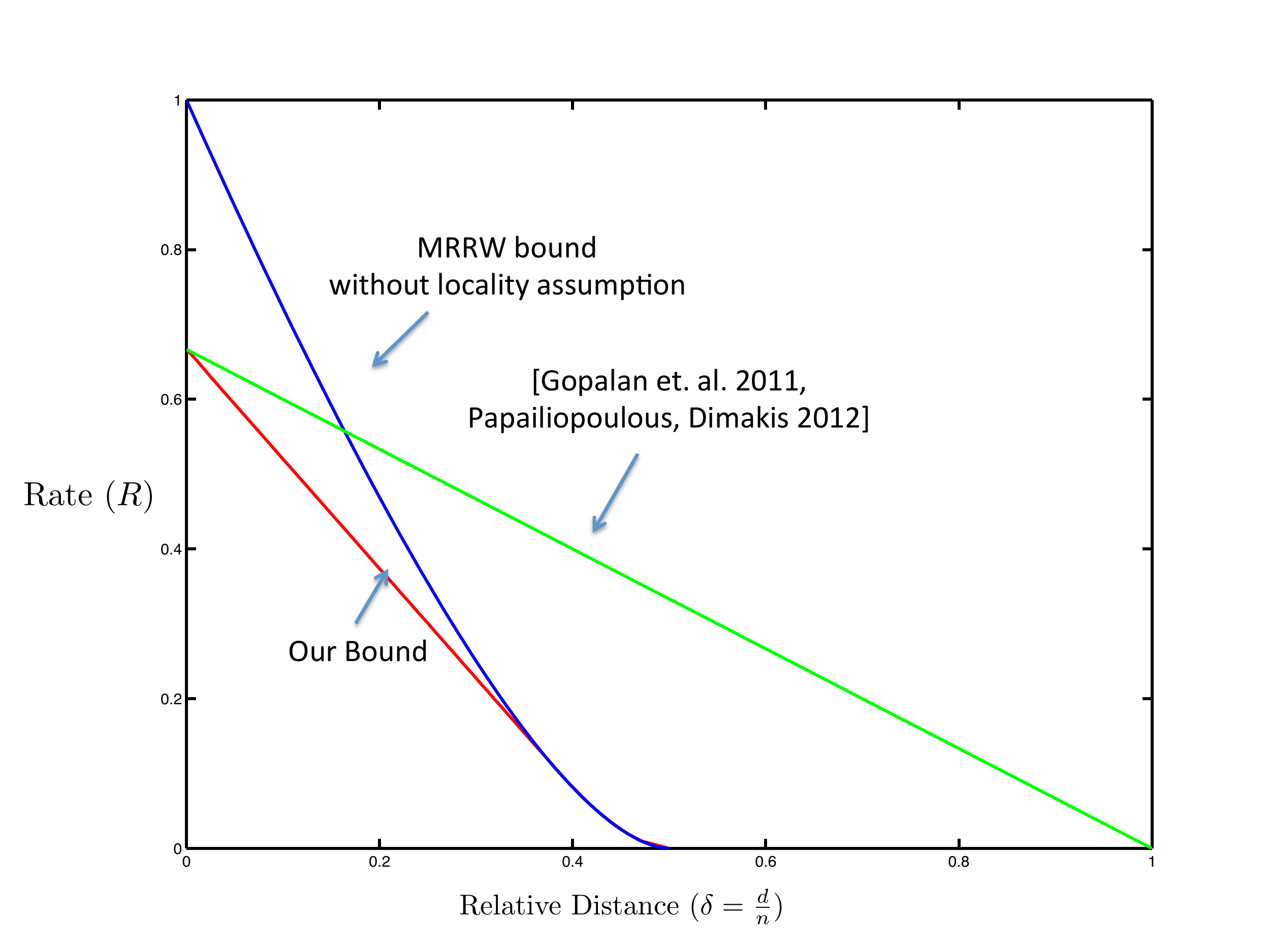, width=3.5 in}}
\caption{A depiction of our bound through the trade-off between the \emph{rate}, $k/n$ and \emph{relative distance}, $d/n,$ for binary codes ($q$=2) for large values of $n,$ with locality $r=2$. The curves plotted are upper bounds on the achievable rates; the plot clearly demonstrates that our upper bound, that uses MRRW bound as a black-box, is better than  the previously known bounds on the rate, for a given relative distance.}
\label{fig:bound}
\end{figure}
\begin{figure}
\center
\centerline{\psfig{figure=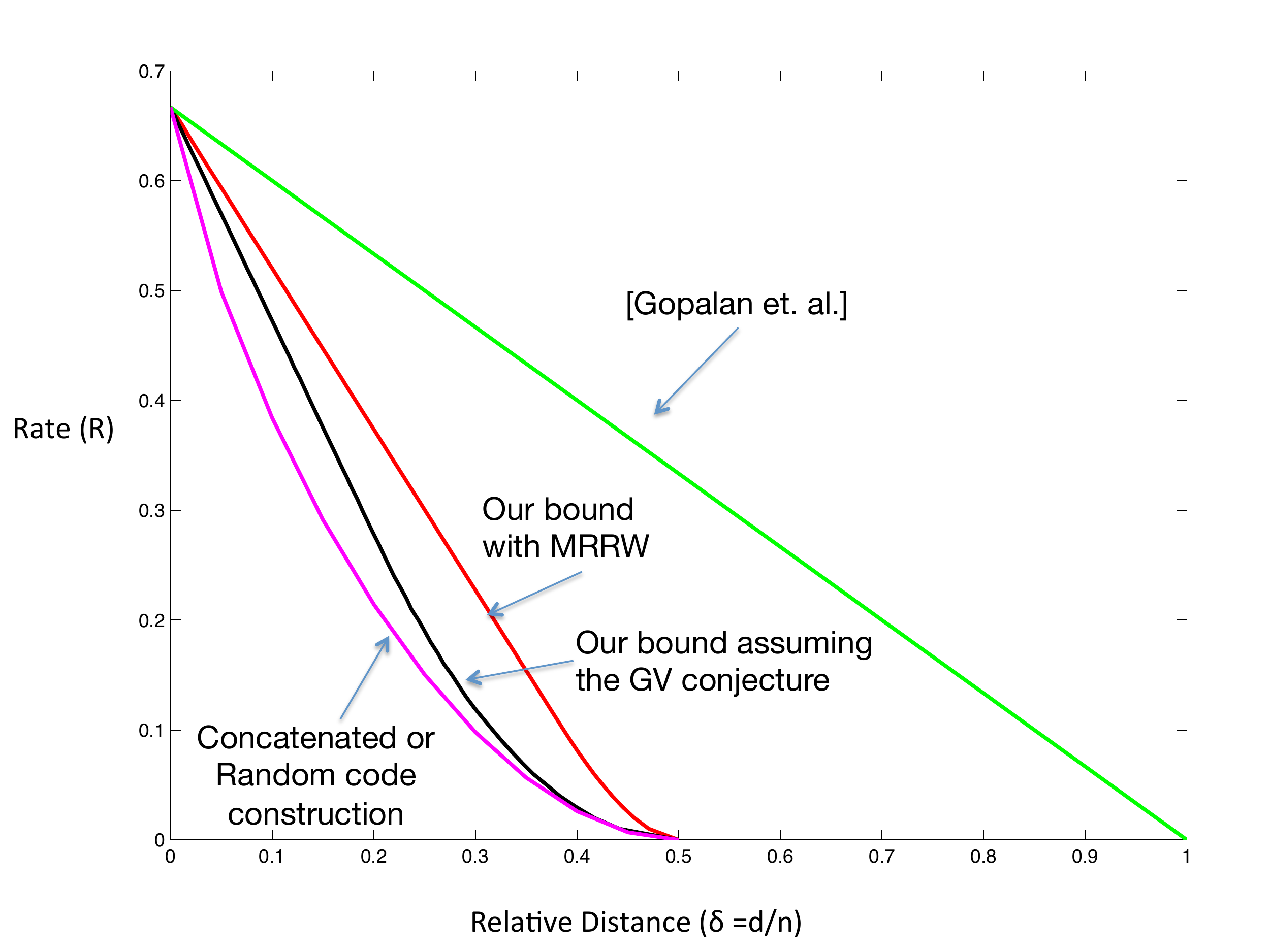, width=3.5 in}}
\caption{A depiction of our achievability result \eqref{eq:achiev} through the trade-off between the \emph{rate}, $k/n$ and \emph{relative distance}, $d/n,$ for binary codes ($q$=2) for large values of $n,$ with locality $r=2$. We compare this achievable rate with our upper bounds: assuming respectively MRRW bound and
the GV bound as the optimal rate for error-correcting codes. If the GV bound were true rate-distance trade-off, then our achievability scheme is very good for large distances. }
\label{fig:bound1}
\end{figure}

The above bound is a generalization of the Singleton bound to include the locality of the codeword, $r$; when $r=k$, the above bound collapses to the classical Singleton bound. In addition, through an invocation of the \emph{multicast} capacity of wireline networks via random network coding, reference \cite{Dimakis_LRC} showed that the bound (\ref{eq:Yekhanin_etal}) is indeed tight for a sufficiently large field size. Intuitively speaking, the bound (\ref{eq:Yekhanin_etal}) implies that there is a \emph{cost} to locality; the smaller the locality, $r,$ the smaller the minimum-distance $d$. Code constructions that achieve the above bound based on Reed-Solomon codes, among other techniques, have been recently discovered in \cite{Dimakis_LRC,VijayKumar_MBRLocality,rawat_LRCoptimal,tamo2013optimal, Sriram_VijayKumar_explicitLRC,huang2007pyramid}\footnote{Recent literature has extended the study of locally recoverable codes to include security constraints \cite{rawat_LRCoptimal}, notions of repair bandwidth \cite{VijayKumar_MBRLocality, Sriram_VijayKumar_explicitLRC}, multiple local repair alternatives \cite{prakash2012optimal}, and probabilistic erasure models \cite{MCW2013} (See \cite{Datta_Oggier_Survey} for a survey). In this work, we concern ourselves with the original notion of locality as described in \cite{Gopalan_Locality,Dimakis_LRC}.}. Missing from these works is a formal study of the impact of an important parameter - the size of the alphabet of the code. Codes over small alphabets are the central subject of  classical coding theory, and are of interest in the application of storage because of their implementation ease. In this paper, we remove the restriction of the large alphabet size from the study of locality of codeword symbols. In particular, we study the impact of the alphabet size on the locality of a code that has a fixed rate.
\subsection{Our Contribution:}
The main contribution of this paper is an upper bound on the minimum distance on the code with a fixed locality that is dependent on the size of the alphabet. While the technical statement of our bound is discussed later (in Theorem \ref{thm:converse}), it is worth noting here that our bound, which is applicable for any feasible alphabet-size and any feasible $(n,k),$ is least as good as the bound of \cite{Gopalan_Locality, Dimakis_LRC} for all parameters. Recall that even in the absence of locality constraints, finding the largest possible minimum distance of a code with a fixed rate over a fixed-size alphabet remains an open problem in general. Our main result uses this quantity - the largest possible minimum distance of an $(n,k)$ code over a given alphabet size - albeit unknown, as a parameter to obtain a bound under locality constraints. As a consequence, our bound is more stringent than the classical (locality-unaware) bounds such as the Mcliece-Rodemich-Rumsey-Welch (MRRW) bounds since they form a special case of our bound (unrestricted locality).  The role of the alphabet size on the rate of the code is highlighted in the plot of Fig. \ref{fig:bound}, where we compare our bounds with existing bounds. Since certain code constructions in previous works are based on multicast codes over networks \cite{Dimakis_LRC}, our result can be interpreted as the demonstration of the impact of alphabet size on the rates of multicast \emph{network codes} for certain networks. Finally, we discuss some achievable constructions in this paper and compare them with our bound. In particular, in Sec. \ref{sec:simplex}, we will show that the family of Simplex codes (dual of Hamming codes) are optimal locally recoverable codes since they meet our new bound. Constructions based on random coding
and concatenated codes are also provided that achieve a rate-distance tradeoff  asymptotically  close to the optimal possible (Section~\ref{sec:achievability}). 
After the initial publication of our work \cite{cadambe_mazumdar_2013}, interesting generalizations of the Simplex code construction of Sec. \ref{sec:simplex} have been studied in \cite{goparaju_calderbank, silberstein_zeh, zeh_yaakobi}. Also worth noting is the elegant construction of an optimal family of  codes by Tamo and Barg \cite{tamo2014family}. Nonetheless, the search for optimal locally recoverable codes remains open even for the binary alphabet. Our randomized construction of Section \ref{sec:achievability} outperform previous code constructions for smaller alphabet sizes including the binary alphabet. %, for certain choices of the parameters $n,k,r,d$.
We end with a discussion on construction based on LDPC-type codes and list decoding.

%In particular, these constructions are asymptotically optimal for low rates ($R \to 0$) as the code length approaches $\infty$ for a fixed locality. We note that the construction of Section \ref{} outperforms the constructions of \cite{} for certain rates, and in particular, is asymptotically optimal for both low rates ($R \to 0$) as well as high rates ($\delta \to 0$). 

\subsubsection*{Notation}
Sets are denoted by calligraphic letters and vectors are denoted by bold font. Consider an element $\mathbf{X} \in \mathcal{A}^{n},$ where $\mathcal{A}$ is an arbitrary finite set. The notation $X_{i} \in \mathcal{A}$ denotes the $i$th co-ordinate of the tuple $\mathbf{X}$. For any set $\mathcal{R} \subseteq \{1,2,\ldots,n\}$, the notation $\mathbf{X}_{\mathcal{R}} \in \mathcal{A}^{|\mathcal{R}|}$ denotes the projection of $\mathbf{X} \in \mathcal{A}^{n}$ on to the co-ordinates corresponding to $\mathcal{R.}$
For $\mathbf{X},\mathbf{Y} \in \mathcal{A}^{n}$, the Hamming distance $\Delta_{H}(\mathbf{X},\mathbf{Y})$ is the cardinality of the set $\{m: X_m \neq Y_m\}.$

\section{System Model: Locally Recoverable Codes}

A code $\mathcal{C}$ with length $n$ over alphabet $\mathcal{Q}$
 consists of $|\mathcal{C}|$ codewords denoted as  $$\mathcal{C}=\{\mathbf{X}^{n}(1),\mathbf{X}^{n}(2),\ldots,\mathbf{X}^{n}(|\mathcal{C}|)\},$$ where $\mathbf{X}^{n}(i) \in \mathcal{Q}^{n}, \forall i$. The \emph{dimension of the code}, denoted by $k$ is defined as $k\define \frac{\log |\mathcal{C}|}{\log|\mathcal{Q}|},$ and the rate of the code denoted as $R$ is defined as $R\define\frac{k}{n}$. An $(n,k,d)$-code over $\mathcal{Q}$ is an $n$ length code $\mathcal{C}$ with dimension $k$ such that the {\em minimum distance} is $d$, i.e., with $$d = \min_{\mathbf{X}^{n},\mathbf{Y}^{n}\in \mathcal{C}, 
\mathbf{X}^{n}\ne \mathbf{Y}^{n} } \Delta_{H}\left(\mathbf{X}^{n}, \mathbf{Y}^{n}\right).$$ We refer to $\delta\define \frac{d}{n}$ as the \emph{relative distance} of the code. 

\begin{definition}\label{def:recovery}
An $(n,k,d)$-code is said to be $r$-\emph{locally recoverable} if for every $i$ such that $1 \leq i \leq n,$  there exists a set $\mathcal{R}_i \subset \{1,2,\ldots,n\}\setminus\{i\}$ with $|\mathcal{R}_i|=r$ such that for any two codewords $\mathbf{X},\mathbf{Y}$ satisfying ${X}_i \neq Y_i,$ we have $\mathbf{X}_{\mathcal{R}_i} \neq \mathbf{Y}_{\mathcal{R}_i}$.
\end{definition}

 Informally speaking, this means that an erasure of the $i$th coordinate
  of the codeword can be recovered by accessing the coordinates associated with $\mathcal{R}_{i}$. 
Hence any erased symbol can be recovered  by probing at most  $r$ other coordinates.

%A code $\cC\in \mathcal{Q}^{n}$ is called recoverable with locality $r$ if
%any symbol of the codeword can be recovered by querying only at most $r$
%other symbols. If $\cC$ has dimension $k= \log_q |\cC|$ and minimum distance $d$
%then it is written a $[n,k,d,r]$-code. The rate $R$ and relateve distance $\delta$ of $\cC$
%is defined as 
%$$ R= k/n\,\,; \quad \delta = d/n.$$
 
%%For the case of linear codes, $\cC$ is locally recoverable with locality $r$ if
%for any column $i, 1\le i\le n,$ of the parity check matrix $H$ of $\cC$, there    
%exists a row $j$, such that $H_{i,j}$ is nonzero, and there are  at most $r$ other
%nonzero entries in that row.

\remove{\section{Achievability}

Let $k^{(q)}_{\rm opt}(n,d)$ denote the maximum dimension attainable by a $q$-ary linear code
of length $n$ and minimum distance $d$. Also, let $R^{(q)}_{\rm opt}(\delta)$ is the optimal
rate given the relative distance $\delta$ as $n$ tends to $\infty$ (assuming the limit exists,
although this fact is not proven).

\begin{theorem}[Gilbert-Varshamov]
\begin{equation}
R^{(q)}_{\rm opt}(\delta) \ge R^{(q)}_{\rm GV}(\delta) \equiv 1 - \hq(\delta),
\end{equation}
where
\begin{equation}
\hq(x)  = x \log_q(q-1) -x\log_q x -(1-x)\log_q(1-x).
\end{equation}
\end{theorem}

\begin{proposition}
There exists $q$-ary $[n,k,d,r]$-codes with
\begin{equation}
k \ge k^{(q)}_{\rm opt}(n,d) - \Big\lceil \frac{n}{r+1} \Big\rceil.
\end{equation}
\end{proposition}
\begin{proof}
We take the parity-check matrix of the optimal code of length
$n$ and distance $d$ and add $\Big\lceil \frac{n}{r+1} \Big\rceil$ rows
to it. All these rows have disjoint support and the size of support of
any row is $r+1$. The loss of dimension of the code is at most 
$\Big\lceil \frac{n}{r+1} \Big\rceil$; however the code has now local recoverability
$r$.
\end{proof}
\begin{corollary}
There exists family of $q$-ary codes with relative minimum distance 
$\delta$ and rate $R^{(q)}_{\rm opt}(\delta) -\epsilon$ that has
local recoverability $O(\frac1\epsilon).$
\end{corollary}

Gallager's existence result regarding LDPC codes gives the following corollary.
\begin{proposition}[Gallager]
There exists family of binary codes with relative minimum distance 
$\delta$ and rate $R^{(2)}_{\rm GV}(\delta) -\epsilon$ that has
local recoverability $O(\log \frac1\epsilon).$
\end{proposition}}

\section{Bound on Minimum Distance for Local Recovery}
Given parameters $n,d,q,$ let
$$k_{\text{opt}}^{(q)}(n,d) = \max \frac{\log|\mathcal{C}|}{\log q},$$
where the maximization is over all possible $n$-length codebooks $\mathcal{C}$ with minimum distance $d$, over some alphabet $\mathcal{Q}$ where $|\mathcal{Q}|=q$. Informally speaking, $k_{\text{opt}}^{(q)}(n,d)$ is the largest possible dimension of an $n$-length code, for a given alphabet size $q$ and a given minimum distance $d$. The determination of $k_{\text{opt}}^{(q)}$  is a classical open problem in coding theory. We also know that $k_\text{opt}$ satisfies the Singleton bound:
$$ k^{(q)}_{\text{opt}}(n,d) \leq n-d+1, \forall q \in \mathbb{Z}_{+}.$$
References \cite{Gopalan_Locality, Dimakis_LRC}, generalized the above bound under locality constraints as Def. \ref{def:recovery}. However, it is well known that the Singleton bound is not tight in general, especially for small values of $q$.  The goal of this paper is to derive a bound on the dimension of an $r$-locally recoverable code in terms of $k_{\text{opt}}^{(q)}.$ Our main result is the following.
\begin{theorem}
For any $(n,k,d)$ code over $\mathcal{Q}$ that is $r$-locally recoverable
\begin{equation}\label{eq:conv}
k \le \min_{ t \in \integers_{+}} \Big[ tr + k^{(q)}_{\rm opt}(n  - t(r+1),d) \Big] ,
%\min_{0\le t \le 
% \lceil n/(r+1) \rceil, t \in \integers} \Big[ tr + k^{(q)}_{\rm opt}(n  - t(r+1),d) \Big] .
%\min\left(\lceil n/(r+1) \rceil, \lceil k/r\rceil\right)} \Big[ tr + k^{(q)}_{\rm opt}(n  - t(r+1),d) \Big] .
\end{equation}
\label{thm:converse}
where $q=|\mathcal{Q}|.$
\end{theorem}
Our bound applies to general (including non-linear) codes, as opposed to only linear codes. Note that, the minimizing value of $t$ in \eqref{eq:conv}, $t^{\ast}$, must satisfy,
$$
t^{\ast} \le \min \Big\{\Big\lceil \frac{n}{r+1}\Big\rceil,\Big\lceil\frac{k}{r}\Big\rceil\Big\}.
$$ 
This is true because,  1) for $t \ge \Big\lceil \frac{n}{r+1}\Big\rceil$, the objective function of the
optimization of \eqref{eq:conv} becomes linearly growing with $t$; 2) for $t \ge \Big\lceil\frac{k}{r}\Big\rceil$,
the right hand side of \eqref{eq:conv} is greater than $k$.

%We denote $\mathbf{X}^{n}(i) = (X_1(i),X_2(i),\ldots,X_n(i))$ and $X_{\mathcal{I}}(i) = \left\{X_j(i):j \in \mathcal{I}\right\}$ where $\mathcal{I}\subseteq \{1,2,\ldots,n\}$. We also use $\mathbf{X}^n, X_j, X_\mathcal{I}$ to denote functions over the message set $\{1,2,\ldots,|\mathcal{Q}|^{k}\}$; when a probability distribution over the message set is assumed, these will also indicate corresponding random variables. Finally, the decoding function is denoted by $\mathbf{X}^{-1}:\mathcal{Q}^{n} \rightarrow \{1,2,\ldots,|\mathcal{Q}|^{k}\}$.
The bound of \cite{Gopalan_Locality, Dimakis_LRC}, i.e. \eqref{eq:Yekhanin_etal}, is weaker than the  bound of Theorem \ref{thm:converse}. To prove this claim, let us show that, if a $(n,k,d,r)$-tuple does not satisfy  \eqref{eq:Yekhanin_etal}, then it will not satisfy \eqref{eq:conv}.

If possible, let the tuple $(n,k,d,r)$ violate \eqref{eq:Yekhanin_etal}, i.e., let
$$
d > n-k -\lceil k/r\rceil + 2.
$$
This sets the following chain of implications.
\begin{align*}
&\min_{ t \in \integers_{+}} \Big[ tr + k^{(q)}_{\rm opt}(n  - t(r+1),d) \Big]\\
&\le  \lfloor (k-1)/r\rfloor r + \max\{n  -  \lfloor (k-1)/r\rfloor (r+1)-d+1,0\}   \\
&=  \max\{n  - \lfloor (k-1)/r\rfloor-d+1, \lfloor (k-1)/r\rfloor r\} \\
&<   \max\{n  - \lfloor (k-1)/r\rfloor - n+k+\lceil k/r\rceil-2 +1,k\}\\
&=\max\{ - \lfloor (k-1)/r\rfloor +k+\lceil k/r\rceil-1,k\}\\
&=k,
\end{align*}
which means \eqref{eq:conv} is not satisfied by this tuple as well.
%immediately follows from the above theorem by merely applying the Singleton bound on $k_{\text{opt}}^{(q)}$ and setting $t=\lceil k/r\rceil.$ Note that, we can set this value of $t$, i.e. $\lceil k/r\rceil \le \lceil n/(r+1)\rceil$, because, if not, then in \eqref{eq:conv} putting $t = \lceil n/(r+1)\rceil$ we have a contradiction.
%\begin{corollary}[\cite{Gopalan_Locality, Dimakis_LRC}]
%For any $r$-locally recoverable $(n,k,d)$ code 
%\begin{equation}\label{eq:old_conv} d \leq n-k - \lceil k/r \rceil + 2.
%\end{equation}
%\end{corollary}

Notice that the above chain of implications came from plugging in the Singleton bound on $k_{\text{opt}}^{(q)}.$ We shall apply bounds that are dependent on $q$ and stronger than the Singleton bound on $k_{\text{opt}}^{(q)}$ to effectively obtain tighter bounds on (\ref{eq:Yekhanin_etal}) later in this paper. We shall first present an overview of the proof of Theorem \ref{thm:converse}.
%It is worth noting that our result implies that the above bound depends on $q;$ for a given value of $q$, any upper bound on the function $k_{\text{opt}}^{(q)}$ can be utilized to obtain a bound for an $r$-locally recoverable code. We demonstrate this using the Plotkin bound later in this section (Section \ref{sec:application}). We next introduce some notation required for the proof of the above theorem.
For purposes of the proof, for a given $n$ length code $\mathcal{C}$ we define the function $H(.)$ as follows
$$H(\mathcal{I})= \frac{\log|\{\mathbf{X}_{\mathcal{I}}: \mathbf{X} \in \mathcal{C}\}|}{\log|\mathcal{Q}|},$$
for any set $\mathcal{I} \subseteq \{1,2,\ldots,n\}$.
\begin{remark}
In the language used in \cite{Dimakis_LRC}, $H(\mathcal{I})$ would denote the ``entropy'' associated with $\mathbf{X}_\mathcal{I}.$ Here, the above definition is appropriate since our modeling is adversarial, i.e., we do not presuppose any distribution on the messages or the codebook (see, \cite{MCW2013} where such assumptions have been made). However, the behavior of the function $H(.)$ is similar to the entropy function; for instance it satisfies submodularity, i.e., $H(\mathcal{I}_{1})+H(\mathcal{I}_{2}) \geq H(\mathcal{I}_{1} \cup \mathcal{I}_{2})+H(\mathcal{I}_{1} \cap \mathcal{I}_{2})$
\end{remark}

Theorem \ref{thm:converse} follows from Lemma \ref{lem:1} and Lemma \ref{lem:2} stated next. 
\begin{lemma}
\label{lem:1}
Consider an $(n,k,d)$-code over alphabet $\mathcal{Q}$ that is $r$-locally recoverable. Then, $\forall~1\leq t \leq k/r, t \in \mathbb{Z}$ there exists a set 
$ ~\mathcal{I} \subseteq \{1,2,\ldots,n\}, |\mathcal{I}| = t(r+1)$ such that ${H}(\mathcal{I}) \leq tr.$ 
\end{lemma}
\begin{lemma} Consider an $(n,k,d)$-code over $\mathcal{Q}$ where there exists a set $\mathcal{I}\in\{1,2,\ldots,n\}$ such that $H({\mathcal{I}})\leq m.$ Then there exists a $(n-|\mathcal{I}|,(k-m)^{+},d)$ code over $\mathcal{Q}$.
\label{lem:2}
\end{lemma}

The above lemmas are proved in the appendix.
%The goal is to construct $\mathcal{I}$ such that there are $t$ elements in $X_\mathcal{I}$ such that each of these elements can be repaired from the other components of $X_\mathcal{I}.$ Note that such a set satisfies $H(X_\mathcal{I}) \leq (|\mathcal{I}|-t) \log(\mathcal{Q})$. The $r$-locality of the code implies that the we can construct such a set with $|\mathcal{I}|=t(r+1)$.

%all $\mathbf{X},\mathbf{Y} \in \hat{C}(\hat{X})$ 

% the codebook as  $\mathcal{C}=\{\mathbf{X}^{n}(1),\mathbf{X}^{n}(2),\ldots,\mathbf{X}^{n}(|\mathcal{C}|)\}$, we aim to show that there exits a "prefix" sequence $$\hat{\mathbf{X}}^|\mathcal{I}| =(\hat{X}(1),\hat{X}(2),\ldots, \hat{X}(|\mathcal{I}|)$$ such that there are  at least $|\mathcal{Q}|^{k-m}$ elements of $\mathcal{C}$ with $\hat{X}$ as a prefix. In other words, there exists a sequence $\hat{\mathbf{X}}^{|\mathcal{I}|}$ such that the cardinality of the set 
%$$\{\mathbf{X} \in \mathcal{C}: \mathbf{X}_{\mathcal{I}} = \hat{\mathbf{X}^{|\mathcal{I}|}$$ is at least $\mathcal{Q}^{k-m}$. If we show this, then the codebook for the new $(n-|\mathcal{I}|,k-m, d)$ code is chosen 

%Denoting $W$ to be the message ( uniformly distributed over the message set), we have
%$$H(W| X_\mathcal{I})  = H(\mathbf{X}^n|{X}_{\mathcal{I}}) = H(\mathbf{X}^n)-H(X_{\mathcal{I}}) \geq (k - m)\log|\mathcal{Q}|$$ %This means that there exists a sequence $x_{\mathcal{I}} \in  \mathcal{Q}^{|\mathcal{I}|}$.
%%$$H(X_{\mathcal{I}^{c}}| X_\mathcal{I})  = H(X)-H(X_{I}^{c}) \geq (k - m)\log|\mathcal{Q}|$$ 
\vspace{0.2in}

\section{Applications of Thm.~\ref{thm:converse}}
In this section, we apply classical bounds for $k_{\text{opt}}$ to Theorem \ref{thm:converse}. To enable a clean analysis, we look at the regime where $n \to \infty.$ In particular we set $R=k/n, \delta = d/n$ and obtain bounds on the trade-off between $(R,\delta)$ as $r$ is fixed and $n \to \infty$. We first apply the Plotkin bound on $k_{\text{opt}}$ and obtain an analytical characterization of the $(R,\delta)$ trade-off with dependence on the alphabet-size, $q$; in particular, we demonstrate a \emph{distance-expansion} penalty as a result of the limit on alphabet size. To obtain a tighter locality-aware bound, we then use the MRRW bound for $k_{\text{opt}}$ to numerically obtain the plot of Fig. \ref{fig:bound}.

To begin, observe that dividing the Singleton bound $n$ and letting $n \to \infty$, it can be written as 
$$ R \leq 1 - \delta+o(1)$$
Similarly, the bound of \cite{Gopalan_Locality, Dimakis_LRC} can be written as:
$$ \delta \leq 1 - \frac{rR}{r+1}+o(1).$$
\begin{equation} \Rightarrow R \leq \frac{r}{r+1}(1 - \delta)+o(1)
\label{eq:1}
\end{equation}
The plot of the above bound is placed in Fig. \ref{fig:bound} for $r=2$. The \emph{cost} of the locality limit above therefore is the factor of $r/(r+1)$ over the Singleton bound. We are now ready to analyze the Plotkin Bound, adapted to Theorem \ref{thm:converse}.
\subsection*{Application of Plotkin Bound - Distance Expansion Penalty}
Let us choose $t = \frac{1}{r+1}(n-\frac{d}{1-1/q})$ in Theorem \ref{thm:converse}. We have,
for any $(n,k,d)$-code that is $r$-locally recoverable,
$$
k \le  \frac{r}{r+1}\Big(n-\frac{d}{1-1/q}\Big) + k^{(q)}_{\rm opt}\Big(\frac{d}{1-1/q},d\Big)
$$
 It is known,  from the Plotkin bound, $k^{(q)}_{\rm opt}\Big(\frac{d}{1-1/q},d\Big) \le \log_q \frac{2qd}{1-1/q}.$
 See, for example, Sec. 2\S2 of MacWilliams and Sloane \cite{MS1977}, for a proof of this result for
 $q=2$, which can be easily extended for larger alphabets.
 Hence,
\begin{equation}\label{eq:plotkin}
k \le \frac{r}{r+1}\Big(n-\frac{d}{1-1/q}\Big) + \log_q \frac{2qd}{1-1/q}.
\end{equation}
Generally, this bound is better than \eqref{eq:Yekhanin_etal}. Notice that dividing the above by $n$ and taking $n \to \infty,$ we have
$$
R = \frac{k}{n} \le  \frac{r}{r+1}\Big(1-\frac{\delta}{1-1/q}\Big) +o(1),
$$
whereas, 
observing the above, it can be noted that the effect of restricting $q$ leads to a distance-expansion penalty of $\frac{1}{1-1/q}$, since the above bound is tantamount to shooting for a distance of $\delta/(1-1/q)$ w.r.t. \eqref{eq:1}. %For instance, if $q=2,$ then our trade-off becomes $R \leq \frac{r}{r+1}(1-2\delta)$; for thi
\subsection*{Beyond the Plotkin bound}
Recall that the MRRW bound is the tightest known bound for the rate-distance tradeoff in absence of locality constraints. We briefly describe an application of this bound for Theorem \ref{thm:converse}, i.e., when the locality is restricted to be equal to a number $r$; it is this bound that is plotted in Fig. \ref{fig:bound}. We restrict our attention to binary codes ($q=2$) and therefore the dependence on $q$ is dropped in the notation.

Define $R_{\text{opt}}(\delta) \define \lim_{n \to \infty} \frac{k_{\text{opt}}(n,\delta n)}{n} $ Dividing the bound of Theorem \ref{thm:converse} by $n$ we can get, as $n \to \infty$,
\begin{eqnarray} R \leq \min_{0 \leq x \leq r/(r+1)} x + \left(1-x(1+1/{r})\right)R_{\text{opt}}\left(\frac{\delta}{1-x(1+1/r)}\right) \label{eq:asymptoticonverse}\end{eqnarray}
where $x = tr/n$.
It is instructive to observe that, setting $x=0$ above yields classical (locality-unaware) bounds. Setting $x=R$ above and writing out the Singleton bound for $R_{\text{opt}}$ yields the bound of (\ref{eq:1}). Therefore the above bound is superior to all the classical (locality-unaware) bounds on $R(\delta)$ and the bound of (\ref{eq:1}) since these are special cases.  %Setting $x=R$ gives us $R_{\text{opt}}(\frac{\delta
Using the MRRW bound $R(y) \leq \mathrm{H}_{2}(0.5 - \sqrt{y(1-y)})+o(1)$ (where $\hq(x) = x\log_q(q-1) -x\log_q x -(1-x)\log_q(1-x)$ represents the $q$-ary entropy function), and numerically solving the optimization problem above (in a brute-force manner) yields our bounds for the rate-distance trade-offs for any given $r$. Deriving analytical insights for the optimization problem by application of bounds beyond the Plotkin bound is an area of future work.
\begin{remark}
While the MRRW bound is the best known upper bound on the rate given a relative distance, for binary codes, the best known achievable scheme, asymptotically as $n \to \infty,$ is given by the Gilbert-Varshamov (GV) Bound. Indeed, it is a folklore conjecture in coding theory that the GV bound is the best achievable rate for binary code, asymptotically as the blocklength tends to infinity. Therefore, to evaluate the merit of binary locally recoverable code constructions, the use of the GV bound for the function $R_{\text{opt}}$ in (\ref{eq:asymptoticonverse}) has operational meaning.
\end{remark}

\section{Simplex codes and tightness of Thm.~\ref{thm:converse}}
\label{sec:simplex}
For alphabet size exponential in the blocklength the bound of \eqref{eq:Yekhanin_etal} has been shown to be achievable in
\cite{rawat_LRCoptimal,tamo2013optimal} by constructing explicit codes. Furthermore, recently, \cite{tamo2014family} has shown that an alphabet that is linear in the blocklength suffices to achieve the bound of \ref{eq:Yekhanin_etal}. Hence  \eqref{eq:conv} is tight for  large alphabets.
We will show that this bound is also achievable for small, in particular binary, alphabets by giving an example of explicit family of codes
where \eqref{eq:conv} is met with equality. The family of codes is $[2^m-1, m, 2^{m-1}]$ Simplex code, $m \in \integers_{+}$.

First, we derive, according to Thm.~\ref{thm:converse}, the best possible locality a code with the parameters of
Simplex code can have.  Here, $n \equiv 2^m-1,\, k \equiv m = \log_2(n+1), \, d \equiv 2^{m-1} = \frac{n+1}2.$ 
We use $t = 2$, which satisfies 
$$
t \le \min \Big\{\Big\lceil \frac{n}{r+1}\Big\rceil,\Big\lceil\frac{k}{r}\Big\rceil\Big\},
$$ 
as will be clear next.
With this value and using Plotkin bound \cite[Sec.~2\S2]{MS1977}:
\begin{align*}
k=m  &\le 2r +k_{\rm opt}(2^m-1 - 2(r+1),2^{m-1})\\
&\le 2r+ \log_2 \frac{2\cdot2^{m-1}}{2^m-2^m+1+2(r+1)}\\
& =2r+m - \log_2 (2r+3) .
\end{align*}
That is, 
$$
2r \ge \log_2 (2r+3) \Rightarrow r \ge 2.
$$
Hence according to Thm.~\ref{thm:converse}, the
best possible locality with the parameters of Simplex code is $2$.

Next, we show that the Simplex code indeed has locality $2$. This is shown by constructing a
parity-check matrix of Simplex code, with every row having exactly $3$ ones. Recall, the dual code of Simplex code
is a $[2^m-1, 2^m-1-m, 3]$-Hamming code. We give a generator matrix of Hamming code
that has only $3$ ones per row.  Let us index the columns of the generator matrix by $1, 2, \dots, 2^m-1$, and use
the notation $(i,j,k)$ to denote the vector with exactly three $1$'s, located at
positions $i, j,$ and $k$. Then, the Hamming code has a generator matrix given by
the row vectors
$(i, 2^j, i+2^j)$ for $1\le j \le m-1, 1\le i < 2^j.$

This gives the first example of a family of algebraic codes that are
optimal in terms of local repairability.

\section{Achievability bounds and constructions}
\label{sec:achievability}
So far in this paper, we have provided upper bounds on on the rate achievable for a fixed locality, distance, and alphabet size. Constructions of locally recoverable codes is an interesting open question especially relevant to practice. To understand the related issues (briefly), consider the special case of binary codes ($q=2$) where, in absence of locality constraints, the best known simple achievable scheme comes via the Gilbert-Varshamov (GV) bound: $R \geq 1-\mathrm{H}_{2}(\delta)$. Note that, to achieve a locality of $r$ with a linear code, it is sufficient for the parity check matrix of the code to have the following property: for every column, there exists a row vector in the parity check matrix with a non-zero entry in that column, and a hamming weight that is no bigger than $r+1$.  A simple construction for locally recoverable codes is constructed by taking the parity check matrix of a code that achieves the GV bound and add $\lceil\frac{n}{r+1}\rceil$ rows to it; each new row has $r+1$ nonzero values and the support of all the (new) rows are disjoint. Clearly, this code has a locality of $r$. Note that this new code has rate:
$$ R \geq 1 - \mathrm{H}_{2}(\delta) - \frac{1}{r+1} = \frac{r}{r+1}-\mathrm{H}_{2}(\delta).$$
For $\delta=0,$ the above clearly meets the outer bound of (\ref{eq:Yekhanin_etal}). However, the above achievable scheme does not meet our bound for larger values of $\delta$. For example, in the regime of Fig. \ref{fig:bound}, i.e., $r=2,$ the above bound implies that $R=0$ for $\delta \geq \mathrm{H}_{2}^{-1}(2/3) \approx 0.18$. Clearly, this is not tight with our bound, where $R > 0$ as long as $\delta < 0.5$. This motivates the following question: what is largest possible (relative) distance of a code with non-zero rate, for a fixed locality and alphabet size? We answer this question in the next section. In particular, we provide two families code constructions that perform well from the perspective of our bounds.
%One of the most promising avenue to pursue to construct a locally recoverable code is to consider an LDPC or low density parity check matrix code. In an LDPC code the rows of the parity check matrix have small (constant) number of non zero values. Clearly, the locality of a linear code is upper bounded by the maximum number of nonzero values in any row of the parity check matrix of the code less one. Hence, if one can construct an LDPC code with a specified rate-distance trade-off, that would be a code with small locality value.

\subsection{Random codes}
Suppose $r+1$ divides $n$. Construct a random code of length $n$ in the following way.
Let $X_{i,j}, 1 \le i \le \frac{n}{r+1}, 1\le j \le r,$ are randomly and uniformly chosen from $\ff_q$. 
Let $X_{i, r+1} = \sum_{j=1}^{r} X_{i,j},$ where the addition is over $\ff_q$.

Assume, $X_{i,j}, 1 \le i \le \frac{n}{r+1}, 1\le j \le r+1,$ is a codeword of a random code. We choose such a random code consisting of
$M$ independent codewords $\mathbf{X}_1,\dots, \mathbf{X}_M$. The length, locality and dimension of any code
in this random ensemble is $n, r$ and $k = \log_q M$ respectively.
\begin{theorem}\label{thm:random}
There exists codes in the above ensemble with
minimum distance at least $d$, where $d$ is given by 
\begin{equation}\label{eq:achiev}
\frac{k}{n} = 1 - \max_{0\le x \le 1}  \Big[  \log_q(1+x(q-1))+ {\frac{1}{r+1}} \log_q \Big(1+(q-1)\Big(\frac{1-x}{1+x(q-1)}\Big)^{r+1}\Big)  -\frac{d}{n}\log_q x\Big].
\end{equation}
\end{theorem}

The proof of this theorem follows the usual random coding methods and the calculation is quite similar to that of the following Theorem~\ref{thm:concat}, that
proclaims the same result for a linear code ensemble. We delegate the proof to the appendix.

\subsection{Concatenated codes}
One approach to construct a locally recoverable code over alphabet size $q$ much smaller than blocklength $n$ is 
 to use Forney's concatenated codes \cite{forney1966concatenated}.
 
 Consider a concatenated code with an outer extended Reed-Solomon code over alphabet $\ff_{q^r}$, length $n_o= q^r$, and
 dimension $k_o$. The minimum distance of the code is $d_o = q^r -k_o +1$. 
 The inner code is a simple $q$-ary parity check code of length $r+1$ (i.e., dimension $r$). The overall code has length $n = (r+1)q^r$,
 dimension $k = k_or$ and distance 
 \begin{align*}
 d &= 2(q^r -k_o+1)\\
 &= 2\Big(\frac{n}{r+1} - \frac{k}{r}+1\Big)\\
 \Rightarrow \quad \frac{k}{n} &= \frac{r}{r+1} - \frac{r}{n} \Big(\frac{d}{2}-1\Big).
 \end{align*} 
Comparing with \eqref{eq:plotkin}, we conclude that this construction has some merit for small values of $r$.

Using concatenated code with a \emph{random} linear outer code, we show a much tighter achievability result. Indeed, the following is true.
\begin{theorem}\label{thm:concat}
There exists an infinite family of $[n,k,d]_q$ concatenated codes with locality $r$, such that \eqref{eq:achiev}
is satisfied.
\end{theorem}
\begin{IEEEproof}
We use an outer random $q^r$-ary linear code of length $\frac{n}{r+1}$ and dimension $\frac{k}{r}$. The inner code is a $q$-ary single
parity-check code of length $r+1$. The overall $q$-ary code has length $n$, dimension $k$ and locality $r$.

For the encoding procedure, any vector in $\ff_q^k \setminus \{0\}$ is first mapped to a vector in $\ff_{q^r}^{k/r}\setminus\{0\}$ and
then encoded to a codeword of the outer code. In the next step, the symbols of the codeword (of the outer code) are mapped to codewords
of the inner code. 

Because the outer code is random linear, for any $\bfu \in \ff_q^k \setminus \{0\}$, the corresponding $q$-ary codeword is
going to have Hamming weight $W$, with,
$$
W = X_1 +X_2 +\dots+X_{\frac{n}{r+1}},
$$
where $X_i \sim X$ are independent identical random variables  such that 
\begin{equation}
\Pr(X = j)=
\begin{cases}
\frac{1}{q^r}\binom{r+1}{j}\frac{q-1}{q} \Big((q-1)^{j-1} +1\Big)& \text{ for  even } j \\
\frac{1}{q^r}\binom{r+1}{j}\frac{q-1}{q} \Big((q-1)^{j-1} - 1\Big)& \text{ for  odd } j.
\end{cases}
\end{equation}
 We used the weight distribution of
$q$-ary single parity-check code from \cite[E.g.~4.6]{Roth_Book}. 
Similar reasoning has been followed in  \cite[Prop.~1]{barg2001concatenated} where more general inner codes were considered. It is instructive to note that when $q=2$, the above equation implies that all even weight codewords are equiprobable, and the odd weight codewords have zero probability.
Now, for any $t>0$,
\begin{align*}
\avg e^{-t X} & =  \frac{q-1}{q^{r+1}} \sum_{j =0}^{r+1} e^{-t j}\binom{r+1}{j} \Big((q-1)^{j-1} +(-1)^j\Big)\\
& = \frac{1}{q^{r+1}}\Big((1+e^{-t}(q-1))^{r+1} + (q-1)(1-e^{-t})^{r+1}\Big).
\end{align*}

Evidently, for any $t>0$,
\begin{align*}
\Pr(W < d) &= \Pr(\sum_{i=1}^{n/(r+1)} X_i  <d)\\
& =\Pr(e^{-t\sum_{i=1}^{n/(r+1)} X_i}  > e^{-td})\\
&\le e^{td} (\avg e^{-t X})^{\frac{n}{r+1}}.
\end{align*}
%We set\footnote{We could have optimized over $t$ to find a slightly better achievability result.} $t = \ln \frac{(q-1)(n-d)}{d}$.

Therefore, the average number of codewords of weight less than $d$ is at most
\begin{align*}
\min_{0\le t}q^k e^{td} (\avg e^{-t X})^{\frac{n}{r+1}}&  =  q^{k-n} \min_{0\le t} e^{td} \Big((1+e^{-t}(q-1))^{r+1} + (q-1)(1-e^{-t})^{r+1}\Big)^{n/(r+1)}\\
& = q^{k-n} \min_{0\le t} e^{td}(1+e^{-t}(q-1))^n\Big(1+(q-1)\Big(\frac{1-e^{-t}}{1+e^{-t}(q-1)}\Big)^{r+1}\Big)^{\frac{n}{r+1}}\\
& =  q^{k-n} \min_{0\le x\le 1} x^{-d}(1+x(q-1))^n\Big(1+(q-1)\Big(\frac{1-x}{1+x(q-1)}\Big)^{r+1}\Big)^{\frac{n}{r+1}}.
%& = q^{k-n} (q-1)^d \Big(\frac{n-d}{d}\Big)^d \Big(\frac{n}{n-d}\Big)^n \Big(1+(q-1)\Big(\frac{1-\frac{d}{(n-d)(q-1)}}{\frac{n}{n-d}}\Big)^{r+1}\Big)^{\frac{n}{r+1}} \\
%&= q^{k-n + n\hq\big(\frac{d}{n}\big)} \Big(1+(q-1)\Big(1- \frac{d}{n}\cdot\frac{q}{q-1} \Big)^{r+1}\Big)^{\frac{n}{r+1}}.
\end{align*}
As long as this number is less than $1$, we must have a code in our ensemble that
has minimum distance at least $d$. This proves the theorem. %Hence, there exist an $[n,k,d]_q$ code with locality $r$,
%\textcolor{red}{Incomplete}
\end{IEEEproof}

It is not immediately apparent as to how tight the bound of \eqref{eq:achiev} is. To see this, let us substitute $x = \frac{d}{(q-1)(n-d)}$, to have
\begin{equation}
\frac{k}{n} = 1-\hq\big(\frac{d}{n}) - \frac{1}{r+1}\log_q \Big(1+(q-1)\Big(1- \frac{d}{n}\cdot\frac{q}{q-1} \Big)^{r+1}\Big).
%\frac{r}{r+1}- \min_{0\le x <q-1}\Big[ \frac{1}{(r+1)\ln q}\ln\Big((1+x)^{r+1}+(1-x)^{r+1}\Big)+\frac{d}{n\ln q}\ln \frac{q-1}{x}\Big]
\end{equation} 
It is clear that at $d=0$,  $\frac{k}{n} =\frac{r}{r+1}$ and at $d = n(1-1/q)$, $\frac{k}{n} =0$. At least at these two points, thus,
the bound of \eqref{eq:achiev} exactly matches the upper bound of \eqref{eq:conv}. We have plotted the bound of \eqref{eq:achiev} by numerically optimizing over the parameter $x$ in Figure~\ref{fig:bound1}. In Figure~\ref{fig:bound1}, we compare this achievability result with our upper bound, assuming the GV conjecture (that the GV bound is the asymptotically optimal achievable rate for a binary error-correcting code). 

\section{Discussions}
\label{sec:discussions}

\subsection{LDPC codes}
As  {\em low density parity check matrix} (LDPC) codes are by definition locally repairable -  any construction
of LDPC codes provides locally recoverable codes. In particular, it is to be noted that there are a number of ensembles of
LDPC codes, either based on random graphs or expander graphs, that have been extensively analyzed for their rate and
distance trade-off \cite{barg2008weight}. It has been observed by Gallager \cite{gallager1962low} and others that the ensemble average distance of such codes approach
the Gilbert-Varshamov bound as the degree of the parity-check graph grows. 
These codes also guarantee multiple recovering sets for each codeword symbol. 

As an example, we can take the ensemble of hypergraph codes with Hamming codes as component codes 
from \cite[Thm.~4]{barg2008weight} and the instances presented therein (see also, \cite{lentmaier1999generalized,boutros1999generalized}).
Using a $[15,11,3]$ Hamming code as local code and a hypergraph with $3$ parts, we are able to construct a code of rate $0.2$,
relative distance $0.2307$ (the GV relative distance for this rate is 0.2430), local repairability $r =11$ and three repair groups for
each symbol. At the same time this codes support cooperative local repair \cite{rawat2014cooperative}, that is, there can be
at most two erasures per repair group, that can still be locally corrected.

\subsection{List decoding}
Our upper-bounding techniques can be extended to bound the local recoverability of a list decodable code. 
Let $A^{(q)}_L(n,s)$ be the maximum possible size of a code $\cC$ such that for any ball of radius $s$ in $\ff_q^n$,
there exist at most $L$ codewords of $\cC$. Such codes are called $(s,L)$-list decodable codes.

\begin{theorem}
Let $\cC$ be an $q$-ary $(s,L)$-list decodable code with length $n$, dimension $k$ and local repairability $r$. Then 
\begin{equation}
k \le \min_{ t \in \integers_{+}} \Big[ tr + k^{(q)}_{L}(n  - t(r+1),s) \Big] ,
\end{equation}
where $k^{(q)}_{L}(n,s) = \log_q A^{(q)}_L(n,s)$, and the minimizing value of $t$ in \eqref{eq:conv}, $t^{\ast}$, must satisfy,
$$
t^{\ast} \le \min \Big\{\Big\lceil \frac{n}{r+1}\Big\rceil,\Big\lceil\frac{k}{r}\Big\rceil\Big\}.
$$ 
\end{theorem}
The proof of this theorem follows that of Theorem \ref{thm:converse}. It is known that, for fixed $q$, if $n$ and
$L$ go to infinity, then  $\frac{k^{(q)}_{L}(n,\sigma n)}{n} \to 1- \hq(\sigma)$.
Therefore, the rate $R$ of an locally recoverable $(\sigma n, L)$-list decodable code must satisfy,
\begin{align*}
R &\le \min_{0\le x\le \frac{r}{r+1}} x +  \left(1-x(1+1/{r})\right)\Big(1-\hq\left(\frac{\delta}{1-x(1+1/r)}\right)\Big) \\
& = 1- \max_{0\le x\le \frac{r}{r+1}} \Big[\frac{x}{r} +  \left(1-x(1+1/{r})\right)\hq\left(\frac{\delta}{1-x(1+1/r)}\right)\Big].
\end{align*}
%The bound is plotted in Figure for  purpose of comparison.
%\section{Acknowledgement}
%This work was partially supported by a start-up grant from University of Minnesota.
%
\bibliographystyle{abbrv}
\bibliography{Thesis,aryabib}

\begin{appendix}
\subsection{Proof of Lemma \ref{lem:1}}
Consider an $r$-locally recoverable $(n,k,d)$-code. For any $i \in \{1,2,\ldots,n\},$ let $\mathcal{R}_{i}$ denote the corresponding repair-set; by definition $|\mathcal{R}_{i}| = r$. The key idea is to construct a set $\mathcal{I}$ having the desired properties. Note that the proof is trivial for $t=1$ since any codeword symbol in combination with the $r$ symbols that form its local repair set form a valid choice of the set $\mathcal{I}$. The construction of the set $\mathcal{I}$ is more challenging for $t > 1$.  Our construction is essentially similar to \cite{Dimakis_LRC}; we describe our construction here for completeness. We choose 
$$\mathcal{I} = \left(\bigcup_{l=1}^{t} \{a_l\}\cup \mathcal{R}_{a_l} \cup \mathcal{S}_l\right) $$
where $a_1,a_2,\ldots, a_t \in \{1,2,\ldots,n\}$ and $\mathcal{S}_{l} \subset \{1,2,\ldots,n\},l=1,2,\ldots, t$ are chosen as follows: 

\vspace{0.1in}

\begin{enumerate}
\item[Begin] Choose $a_1$ arbitrarily from $\{1,2,\ldots,n\}$. Choose $\mathcal{S}_{1}$ to be the null set.
\item[Loop] For $m=2$ to $m=t$\\
\begin{itemize}
	\item[Step 1:] Choose $a_m$ so that $$a_m \notin \bigcup_{l=1}^{m-1} \{a_l\}\cup \mathcal{R}_{a_l}\cup \mathcal{S}_{l}$$
	\item[Step 2:] Let $\mathcal{I}_{m-1} = \bigcup_{l=1}^{m-1} \{a_{l}\}\cup \mathcal{R}_{a_{l}} \cup \mathcal{S}_{l}.$ Choose $\mathcal{S}_{m}$ to be set of $m(r+1)-\left|\{a_{m}\}\cup \mathcal{R}_{a_{m}} \cup \mathcal{I}_{m-1}\right|$ elements, arbitrarily from $\{1,2,\ldots,n\}- \{a_{m}\}\cup \mathcal{R}_{a_{m}}\cup \mathcal{I}_{m-1}$.
\end{itemize}
\item[End]
%\item[Termination:] Choose $\mathcal{I}_0$ to be an aribtrary subset of $\{1,2,\ldots,n\}- \bigcup_{l=1}^{m-1} \{a_l\}\cup \mathcal{R}_{a_l}$ having cardinality $tr - \left|\bigcup_{l=1}^{m-1} \{a_l\}\cup \mathcal{R}_{a_l}\right|$.
\end{enumerate}
This completes the construction. Note that $\mathcal{I}$ constructed above has cardinality $t(r+1)$. It remains to show that $H(\mathcal{I}) \leq tr$. We now intend to show that $H(\mathcal{I}) = H(\mathcal{I}-\{a_1, a_2,\ldots, a_t\})$ from which the desired bound would follow because of
$$H(\mathcal{I}) = H(\mathcal{I}-\{a_1,a_2,\ldots, a_t\}) \leq t(r+1)-t = tr,$$
where we have used the fact that $H(\mathcal{A}) \leq |\mathcal{A}|$ for any set $\mathcal{A}$.
%\item $H({\{a_1, a_2,\ldots, a_t\} \cup \mathcal{A}}) = H(\mathcal{A})$ if $\bigcup_{i=1}^{t} \mathcal{R}_{a_i} \subseteq \mathcal{A}$. 
%\begin{enumerate}
%\item $H({}\cup \mathcal{R}_{a}\cup \mathcal{A}) = H(\mathcal{R}_{a}\cup \mathcal{A}), \forall a=1,2,\ldots, n$ for any set $\mathcal{A} \subseteq \{1,2,\ldots,n\}$. 
%\item For any $\mathcal{B} \subseteq \mathcal{A} \subseteq \{1,2,\ldots,n\}$, $$H(\mathcal{A}) - |\mathcal{A}| \leq H(\mathcal{B})-|\mathcal{B}|.$$ 
%In particular, $H(\mathcal{A}) \leq |\mathcal{A}|$ for all $\mathcal{A}$.
%\item  $H(\mathcal{A} \cup \mathcal{B}) \leq H(|\mathcal{A}|)+H(|\mathcal{B}|)$ for any sets $\mathcal{A}, \mathcal{B} \subseteq \{1,2,\ldots,n\}.$ 
%\end{enumerate} 
%If $\bigcup_{i=1}^{t} \mathcal{R}_{a_i} \subseteq \mathcal{A}$, the definition of locality readily implies that there is a one-to-one mapping between $\mathbf{X}_{\{a_1,a_2\ldots, a_t\} \cup \mathcal{A}}$ and $\mathbf{X}_{\mathcal{A}}$ which implies 1).  2) also follows from the defintion of $H(.)$. Using the above properties, 
We therefore intend to show a one-to-one mapping between $\{\mathbf{X}_{\mathcal{I}-\{a_1, a_2,\ldots, a_t\}}\}$ and $\{\mathbf{X}_{\mathcal{I}}\}.$ In other words, suppose that $ \mathbf{X}_{\mathcal{I}} \neq \mathbf{\hat{X}}_{\mathcal{I}},$ we need to prove that $\mathbf{X}_{\mathcal{I}-\{a_1, a_2,\ldots, a_t\}}\neq \mathbf{\hat{X}}_{\mathcal{I}-\{a_1, a_2,\ldots, a_t\}}$. Equivalently, suppose that $ \mathbf{X}_{\{a_1,a_2,\ldots, a_t\}} \neq \mathbf{\hat{X}}_{\{a_1, a_2, \ldots, a_t\}},$ we need to prove that $\mathbf{X}_{\mathcal{I}-\{a_1, a_2,\ldots, a_t\}}\neq \mathbf{\hat{X}}_{\mathcal{I}-\{a_1, a_2,\ldots, a_t\}}$.  Suppose a contradiction, i.e., suppose that $\exists, \mathbf{X},\mathbf{\hat{X}}\in \mathcal{C}$ such that
$$ \mathbf{X}_{\{a_1, a_2, \ldots, a_t\}} \neq \mathbf{\hat{X}}_{\{a_1, a_2, \ldots, a_t\}} $$
$$ \mathbf{X}_{\mathcal{I}-\{a_1, a_2, \ldots, a_t\}} = \mathbf{\hat{X}}_{\mathcal{I}-\{a_1, a_2, \ldots, a_t\}} $$
Define $\mathcal{B} =\{j:\mathbf{X}_{j}\neq \mathbf{\hat{X}}_{j}, j \in \{a_1, a_2,\ldots, a_t\}\}$. Note that $\mathcal{B} \subseteq \{a_1, a_2, \ldots a_t\}$. Because of the definition of locality and because $\mathcal{R}_{a_i} \in \mathcal{I}$, the above conditions imply that 
\begin{equation} \mathcal{R}_{i}\cap \mathcal{B}\neq \phi, \forall i \in \mathcal{B} 
\label{eq:something}
\end{equation}
In other words, the repair set associated with any element, $i$, in $\mathcal{B}$ should have at least one element in $\mathcal{B}$, because $\mathbf{X}_{j} = \mathbf{\hat{X}}_{j}$ for all $i \neq j, j \in \mathcal{I}-\mathcal{B}$. We will show that this is a contradiction to our construction. In particular, we will throw away elements from $\mathcal{B}$ one at a time to obtain, from (\ref{eq:something}), a relation of the form $j \cap \mathcal{R}_{j} \neq \phi$ for some $j \in \mathcal{B},$ which is a contradiction. To keep the notation clean, we will show the proof for $\mathcal{B}=\{a_1, a_2, \ldots, a_m\}$, where $m = |\mathcal{B}|$. Our idea generalizes for arbitrary $\mathcal{B}$.
By construction (Step 1), note that $a_m \notin \mathcal{R}_{a_i}, i=1,2,\ldots,m-1.$  Therefore, $a_m$ is not a member of the repair sets of any of the elements of $\mathcal{B}$, and (\ref{eq:something}) implies that
$$\mathcal{R}_{i}\cap \{a_1, a_2\ldots, a_{m-1}\} \neq \phi, \forall i \in \{a_1, a_2, \ldots, a_{m-1}\} $$ 
%$$\mathcal{R}_{a_i}\cap \{a_1, a_2\ldots, a_{t-1}\} \neq \phi, \forall i=1,2,\ldots, t-1. $$ 
Similarly, note that ${a}_{m-1} \notin \mathcal{R}_{a_i},i=1,2,\ldots, m-2$ and Therefore, $a_{m-1}$ is not a member of the repair sets of any of the elements of $\mathcal{B}-\{a_m\}$. So we get, 
$$\mathcal{R}_{i}\cap \{a_1, a_2\ldots, a_{m-2}\} \neq \phi, \forall i \in \{a_1, a_2,\ldots, a_{m-2}\} $$ 
%Let $m$ be the smallest element of ${j:\mathbf{X}_{a_j}\neq \mathbf{\hat{X}}_{a_j}\}.$ Repeating the above procedure $m-1$ times, we reach the following contradiction:
Repeating the above procedure $m-1$ times, we get
$$\mathcal{R}_{a_1}\cap \{a_1 \} \neq \phi, $$
which is a contradiction. 

\subsection{Proof of Lemma \ref{lem:2}}
Without loss of generality, let us assume that $\mathcal{I}=\{1,2,\ldots,|\mathcal{I}|\}$. Consider any element $\mathbf{Z}$ of the $\mathcal{S}=\{\mathbf{X}_{\mathcal{I}}: \mathbf{X} \in \mathcal{C}\}.$ Now, notice that the set of all elements of $\mathcal{C}$ which have $\mathbf{Z}$ as a ``prefix'' can be used to construct a codebook $\mathcal{C}(\mathbf{Z})$ of length $(n-|\mathcal{I}|).$ In particular denote $$ \tilde{\mathcal{C}}(\mathbf{Z}) = \{\mathbf{X}_{\{|\mathcal{I}|+1,|\mathcal{I}|+2,\ldots,n\}}: \mathbf{X}_{\mathcal{I}} = \mathbf{Z}\}$$ 
In addition, we can deduce that the codebook $\tilde{\mathcal{C}}(\mathbf{Z})$, has minimum distance $d$. To see this, consider $\mathbf{U},\mathbf{V} \in \tilde{\mathcal{C}}(\mathbf{Z})$ and note that
\begin{equation}
\Delta_H(\mathbf{U},\mathbf{V}) = \Delta_H((\mathbf{Z},\mathbf{U}),(\mathbf{Z},\mathbf{V})) \geq d
\end{equation}
where, above we have used the fact that, by definition of $\tilde{\mathcal{C}}(\mathbf{Z}),$ the tuples $(\mathbf{Z},\mathbf{U})$ and $(\mathbf{Z},\mathbf{V})$ are elements of $\mathcal{C}$ and therefore have a Hamming distance larger than or equal to $d$. Now, all we need to show is that there exists at least one $\hat{\mathbf{Z}} \in \mathcal{S}$ such that the dimension of $\tilde{\mathcal{C}}(\hat{\mathbf{Z}})$ is (at least) as large as $k-m$. This can be shown using an elementary probabilistic counting argument. Specifically, by assuming that $\mathbf{Z}$ is uniformly distributed over $\mathcal{S},$  the average value of $|\tilde{\mathcal{C}}(\mathbf{Z})|$ can be bounded as follows. 
\begin{eqnarray*} 
|\mathcal{C}|=|\mathcal{Q}|^k &=&\sum_{\mathbf{Z} \in \mathcal{S}}|\tilde{\mathcal{C}}(\mathbf{Z})| \\&=& |\mathcal{S}|E\left[\tilde{\mathcal{C}}(\mathbf{Z})\right] \\
\Rightarrow E\left[\tilde{\mathcal{C}}(\mathbf{Z})\right] &=&  \frac{|\mathcal{Q}^{k}|}{|\mathcal{S}|} \\& \geq& \frac{|\mathcal{Q}|^{k}}{|\mathcal{Q}|^{m}} = |\mathcal{Q}|^{k-m}   
%|\mathcal{C}|=|\mathcal{Q}|^k = \sum_{\hat{X} \in \mathcal{S}}|\hat{C}(\hat{X})| &=& |\mathcal{S}|E(\hat{C}(\hat{X}) \\
\end{eqnarray*}
where, above we have used the premise of the lemma, namely $|\mathcal{S}| = |\mathcal{Q}|^{H(\mathcal{I})} \leq |\mathcal{Q}|^{m}$. Therefore, there is at least one $ \hat{\mathbf{Z}} \in \mathcal{S}$ such that $\tilde{\mathcal{C}}(\hat{\mathbf{Z}}) \geq |\mathcal{Q}|^{k-m}$ thereby resulting in a $(n-|\mathcal{I}|,k-m,d)$ codebook over $\mathcal{Q}$. This completes the proof.

\subsection{Proof of Theorem \ref{thm:random}}

For any two randomly chosen codewords, let $W$ be the distance between them. The distribution of $W$ is exactly the same as
provided in the proof of Theorem  \ref{thm:concat}.

More precisely, define for two codewords $\mathbf{X}_i$ and $\mathbf{X}_j$, $1 \le i,j \le M, i \ne j$, the event $\Omega_{i,j} = \{\Delta_H(\mathbf{X}_i,\mathbf{X}_j) < d \}.$ We have,
$$
\Pr(\Omega_{i,j}) = \Pr(W <d).
$$
Consider the dependency graph of the events $\{\Omega_{i,j}\}_{1\le i,j \le M, i \ne j}$. In this graph two vertices corresponding to the events will
have an edge between them if the events are dependent.
This graph has order $M(M-1)$ and degree at most $2(M-1)$. Hence, using Lov\'{a}sz Local Lemma,
$$
\Pr(\cap_{i,j} \bar{\Omega}_{i,j}) >0,
$$ 
as long as $ \Pr(W<d) (2M-1) <\frac{1}{e}.$ Note that, this means the existence of a locally recoverable code (with parameters $n$, $k$, $d$ and $r$)
 in the ensemble as long as
 $\Pr(W<d)  <\frac{1}{e(2M-1)}$. 
Plugging in the value calculated for $\Pr(W <d)$ from the proof of Theorem  \ref{thm:concat}, we arrive at the statement of the theorem.

\end{appendix}

\end{document}